\def\be{\begin{equation}}
\def\ee{\end{equation}}
\def\ba{\begin{array}}
\def\ea{\end{array}}

\documentclass[pra,showpacs,twocolumn,amsmath]{revtex4}
\usepackage{amsfonts}
\usepackage{mathrsfs}
\usepackage{graphicx}
\usepackage{amssymb}
\usepackage{pifont}
\usepackage{epsfig,subfigure,dsfont,amsthm,amsbsy,mathrsfs,amscd}
\usepackage{hyperref}
\def\qed{\leavevmode\unskip\penalty9999 \hbox{}\nobreak\hfill
     \quad\hbox{\leavevmode  \hbox to.77778em{%
               \hfil\vrule   \vbox to.675em%
               {\hrule width.6em\vfil\hrule}\vrule\hfil}}
     \par\vskip3pt}

\newtheorem{theorem}{Theorem}

\newtheorem{lemma}{Lemma}

\newtheorem{example}{Example}
\input amssym.def

\begin{document}
\title{\large\bf Separability and lower bounds of quantum entanglement based on realignment}

\author{ Jiaxin Sun,$^{1}$  Hongmei Yao,$^{1, *}$ Shao-Ming Fei,$^{2, \dag}$ and Zhaobing Fan$^{1}$}

\affiliation{ ${^1}$Department of Mathematics, Harbin Engineering University,  Harbin 150001, P.R. China \\
$^{2}$School of Mathematical Sciences, Capital Normal University, Beijing 100048, China}
%\date{}
\bigskip

\begin{abstract}
The detection and estimation of quantum entanglement are the essential issues in the theory of quantum entanglement. We construct matrices based on the realignment of density matrices and the vectorization of the reduced density matrices, from which a family of separability criteria are presented for both bipartite and multipartite systems. Moreover, new lower bounds of concurrence and convex-roof extended negativity are derived. Criteria are also given to detect the genuine tripartite entanglement. Lower bounds of the concurrence of genuine tripartite entanglement are presented. By detailed examples we show that our results are better than the corresponding ones in identifying and estimating quantum entanglement as well as genuine multipartite entanglement.
\end{abstract}

\pacs{03.67.-a, 02.20.Hj, 03.65.-w} \maketitle

\section{Introduction}
Quantum entanglement plays an essential role in quantum information processing \cite{nielsen2010quantum} such as quantum computation \cite{divincenzo1995quantum}, quantum teleportation \cite{bennett1993teleporting,albeverio2002optimal} and quantum cryptographic schemes \cite{ekert1991quantum,fuchs1997optimal}. However, it is a difficult problem to distinguish entangled states from the separable ones in general \cite{gurvits2003classical}. In the last decades, a variety of separability criteria have been presented to detect entanglement, such as the positive partial transpose (PPT) criterion or Peres-Horodecki criterion \cite{peres1996separability,horodecki1996separability}, computable cross-norm or realignment criteria \cite{rudolph2003some,rudolph2005further,chen2002matrix,zhang2008entanglement}, entanglement witnesses \cite{terhal2000bell,chruscinski2014entanglement} etc., see, e.g., \cite{horodecki2009quantum,guhne2009entanglement} for comprehensive surveys.

The PPT criterion \cite{peres1996separability} says that any partial transposed bipartite separable state is still positive semidefinite. This criterion is sufficient and necessary for the separability of $m \times n$ quantum sates with $m n \leq 6$ \cite{horodecki1996separability}. However, for higher-dimensional cases, there exist $2 \times 4$ and $3 \times 3$ states which are PPT but still entangled, termed as bound entangled states \cite{horodecki1997separability}. Another famous criterion is the realignment criterion \cite{rudolph2003some,rudolph2005further,chen2002matrix,zhang2008entanglement} which detects many bound entangled states. In \cite{lupo2008bipartite,li2011note}, by analyzing the symmetric functions of the singular values of the realigned matrices, the authors presented separability criteria which are better than the realignment criterion. In \cite{zhang2008entanglement}, based on the realignment of the matrix $\rho-\rho_{A} \otimes \rho_{B}$, improved realignment criterion has been presented, which detects entanglement better than the realignment criterion. In \cite{shi2023family}, the authors presented a family of separable criteria for bipartite states, which reduce to the improved realignment criterion for particular cases. In \cite{shen2015separability}, based on the realigned bipartite density matrix, the vectorization of the reduced density matrices and a parameterized Hermitian matrix, a family of separability criteria were presented, which detect entanglement more efficiently than the realignment criterion, and work also for multipartite systems. In \cite{zhang2017realignment}, based on the sequential realignment of density matrices, a separability criterion for multipartite quantum states has been presented. It detects the multipartite entanglement better than the criteria given in \cite{shen2015separability}. In \cite{li2017detection}, based on the realignment criterion and PPT criterion, the authors presented a criterion to detect the genuine multipartite entanglement of multipartite states. In \cite{qi2024detection}, based on the realigned matrix given in \cite{shi2023family}, a separability criterion for multipartite states has been presented, which detects the genuine multipartite entangled (GME) states that
are not separable with respect to any bipartition.

Another important problem in the study of quantum entanglement is the quantification of entanglement. Various entanglement measures have been presented in recent years \cite{horodecki2009quantum,guhne2009entanglement,lee2003convex,chen2005concurrence,de2007lower,huber2013entropy,chen2016lower,zhu2018lower}. The concurrence and the convex-roof extended negativity (CREN) are two well known measures of entanglement. Nevertheless, it is formidably difficult to find the analytical formulae of such entanglement measures in general due to the optimal minimization involved. In \cite{chen2005concurrence}, based on the positive partial transposition and realignment separability criteria, the lower bounds of concurrence were obtained. In \cite{de2007lower} the author presented analytical lower bounds of concurrence in terms of the local uncertainty relations and the correlation matrix separability criteria. In \cite{ma2011measure} the GME concurrence was presented, whose lower bound has been derived in \cite{li2017measure} based on the norms of the correlation tensors. In \cite{li2020improved}, the authors presented improved lower bounds of concurrence and the convex-roof extended negativity \cite{lee2003convex} based on Bloch representations.

In this paper, we first present a family of separability criteria for bipartite systems, from which we derive tighter lower bounds of concurrence and convex-roof extended negativity in Section \ref{S:2}. In Section \ref{S:3}, we generalize our separability criteria to multipartite systems. These criteria detect genuine multipartite entanglement as well as the multipartite full separability. Tighter lower bound of GME concurrence is obtained. By detailed examples we show that our separability criteria and lower bounds are better than the correspondingly existing ones. We summarize and conclude in Section IV.

%\begin{figure}[ptb]
%\includegraphics[width=0.45\textwidth]{non2.eps}\caption{
%A quantum state $\rho_{ABC}^{(1)}$ is initially shared by Alice, Bob and Charlie$^{(1)}$. Charlie$^{(1)}$ performs a measurement on her part and passes it to Charlie$^{(2)}$. The post-measurement state is $\rho_{ABC}^{(2)}$. Charlie$^{(2)}$ measures $\rho_{ABC}^{(2)}$ on her part and passes it to Charlie$^{(3)}$ and so on.}
%\end{figure}

\section{Detection and measures of entanglement for bipartite states}\label{S:2}
\subsection{Separability criteria for bipartite systems}
Let $\mathbb{C}^{m \times n}$ be the set of all $m\times n$ matrices over complex field $\mathbb{C}$, and $\mathbb{R}$ be the real number field. For a matrix $A=\left[a_{i j}\right]\in \mathbb{C}^{m \times n}$, the vectorization of matrix $A$ is defined as
$\operatorname{vec}(A)=\left(a_{11}, \!\cdots\!, a_{m 1}, a_{12}, \!\cdots\!, a_{m 2}, \!\cdots\!, a_{1 n}, \!\cdots\!, a_{m n}\right)^{T}$, where $T$ stands for the transpose.

Let $Z$ be an $m \times m$ block matrix with sub-blocks $Z_{i, j} \in \mathbb{C}^{n \times n}$, $i, j=1, \ldots, m$. The realigned matrix $\mathcal{R}(Z)$ of $Z$ is defined by
\begin{eqnarray}\notag
	\mathcal{R}(Z)=\left(\begin{array}{c}
		\operatorname{vec}\left(Z_{1,1}\right)^{T} \\
		\vdots \\
		\operatorname{vec}\left(Z_{m, 1}\right)^{T} \\
		\vdots \\
		\operatorname{vec}\left(Z_{1, m}\right)^{T} \\
		\vdots \\
		\operatorname{vec}\left(Z_{m, m}\right)^{T}
	\end{array}\right).
\end{eqnarray}

The realignment criterion \cite{chen2002matrix} says that any separable state $\rho$ in $\mathbb{C}^{d_{A}} \otimes \mathbb{C}^{d_{B}}$ satisfies
$\|\mathcal{R}(\rho)\|_{t r} \leqslant 1$, where $\|A\|_{t r}=\operatorname{tr}\left(\sqrt{A^{\dagger} A}\right)$ is the trace norm of $A$.

For any quantum state $\rho$ in $\mathbb{C}^{d_{A}} \otimes \mathbb{C}^{d_{B}}$, we define
\begin{eqnarray}\label{eq:m1}
	\mathcal{M}_{\alpha, \beta}^{l}(\rho)=\left(\begin{array}{cc}
		\alpha \beta E_{l \times l} & \alpha \omega_{l}\left(\operatorname{tr}_{A}(\rho)\right)^{T} \\
		\beta \omega_{l}\left(\operatorname{tr}_{B}(\rho)\right) & \mathcal{R}(\rho)
	\end{array}\right),
\end{eqnarray}
where $\alpha$ and $\beta$ are arbitrary real numbers, $l$ is a natural number, $E_{l \times l}$ is the matrix with all $l \times l$ elements being $1$, $\operatorname{tr}_{A}$ is the partial trace over the subsystem $A$, and $\omega_{l}(X)=\underbrace{(\operatorname{vec}(X), \cdots, \operatorname{vec}(X))}_{l}$ for any complex matrix $X$. Concerning $\mathcal{M}_{\alpha, \beta}^{l}(\rho)$ we have the following lemma, see proof in Appendix A.

\begin{lemma} \label{P:1} $\left(1\right)$ For any $\rho_{i}\in\mathbb{C}^{d_{A}}\!\otimes \mathbb{C}^{d_{B}}$ and $k_{i}\in \mathbb{R} $, $i=1,2, \ldots, n$, such that $\sum\limits_{i=1}^{n} k_{i}=1$, we have
$$
\mathcal{M}_{\alpha, \beta}^{l}\left(\sum_{i=1}^{n} k_{i} \rho_{i}\right)=\sum_{i=1}^{n} k_{i} \mathcal{M}_{\alpha, \beta}^{l}\left(\rho_{i}\right).
$$
$\left(2\right)$ Let $U$ and $V$ be unitary matrices on subsystems of A and B, respectively. For any $\rho\in\mathbb{C}^{d_{A}} \otimes \mathbb{C}^{d_{B}}$, we have $\left\| \mathcal{M}_{\alpha, \beta}^{l}\left((U \otimes V) \rho(U \otimes V)^{\dagger}\right)\left\|_{t r}=\right\| \mathcal{M}_{\alpha, \beta}^{l}\left(\rho\right) \|_{t r}\right.$.
\end{lemma}

By Lemma 1 we have the following separability criterion.

\begin{theorem}\label{th:1} If a state $\rho\in\mathbb{C}^{d_{A}} \otimes \mathbb{C}^{d_{B}}$ is separable, then
\begin{eqnarray}\notag
\left\|\mathcal{M}_{\alpha, \beta}^{l}\left(\rho\right)\right\|_{t r} \leq \sqrt{\left(l \alpha^{2}+1\right)\left(l \beta^{2}+1\right)},
\end{eqnarray}
where $\mathcal{M}_{\alpha,\beta}^{l}\left(\rho\right)$ is defined in (\ref{eq:m1}).
\end{theorem}

\begin{proof}
Since $\rho $ is separable, it can be written as a convex combination of pure states,
$\rho=\sum_{i} p_{i} \rho_{A}^{i} \otimes \rho_{B}^{i}$, where $p_{i} \in[0,1]$ with $\sum\limits_{i} p_{i}=1$, $\rho_{A}^{i}$ and $\rho_{B}^{i}$ are pure states of the subsystems $A$ and $B$, respectively. From Lemma 1 we have
\begin{eqnarray}\label{eq:4}\notag
&& \left\|\mathcal{M}_{\alpha, \beta}^{l}\left(\rho\right)\right\|_{t r}\\\notag&=&\left\|\sum_{i} p_{i} \mathcal{M}_{\alpha, \beta}^{l}\left(\rho_{A}^{i} \otimes \rho_{B}^{i}\right)\right\|_{t r} \\&\leq& \sum_{i} p_{i}\left\|\mathcal{M}_{\alpha, \beta}^{l}\left(\rho_{A}^{i} \otimes \rho_{B}^{i}\right)\right\|_{t r}.
\end{eqnarray}

Since for any $A\in \mathbb{C}^{m \times n}$ and $B \in \mathbb{C}^{p \times q}$,
\begin{eqnarray}\label{eq:r4}
\mathcal{R}(A \otimes B)=\operatorname{vec}(A)\operatorname{vec}(B)^{T},
\end{eqnarray}
we obtain from the definition of $\mathcal{M}_{\alpha,\beta}^{l}\left(\rho\right)$,
	\begin{eqnarray}\label{eq:6}\notag
			& & \left\|\mathcal{M}_{\alpha, \beta}^l\left(\rho_A^i \otimes \rho_B^i\right)\right\|_{t r} \\\notag
			&=&\left\|\left(\begin{array}{cc}
				\alpha \beta E_{l \times l} & \alpha \omega_l\!\left(\!\operatorname{tr}_A\left(\rho_A^i \otimes \rho_B^i \!\right)\!\right)^T \\
				\beta \omega_l\left(\!\operatorname{tr}_B\left(\rho_A^i \otimes \rho_B^i\!\right)\right) & \mathcal{R}\left(\rho_A^i \otimes \rho_B^i\!\right)
			\end{array}\right)\right\|_{t r} \\\notag
			& =&\left\|\left(\begin{array}{cc}
				\alpha \beta E_{l \times l} & \alpha \omega_l\left(\rho_B^i\right)^T \\
				\beta \omega_l\left(\rho_A^i\right) & \operatorname{vec}\left(\rho_A^i\right) \operatorname{vec}\left(\rho_B^i\right)^T
			\end{array}\right)\right\|_{t r} \\
			& =&\|\left(\begin{array}{c}
				\alpha E_{l \times 1} \\
				\operatorname{vec}\left(\rho_A^i\right)
			\end{array}\right)\left(\begin{array}{cc}
				\beta E_{1 \times l} \quad\left.\operatorname{vec}\left(\rho_B^i\right)^T\right)
			\end{array} \|_{t r}.\right.
	\end{eqnarray}
As $\rho_{A}^{i}$ and $\rho_{B}^{i}$ are pure states, i.e., $\operatorname{tr}\left(\rho_A^i\right)^2=\operatorname{tr}\left(\rho_B^i\right)^2=1$, we have
	\begin{eqnarray}\label{eq:27}\notag
			& &\left\|\left(\begin{array}{c}
				\alpha E_{l \times 1} \\
				\operatorname{vec}\left(\rho_A^i\right)
			\end{array}\right)\left(\beta E_{1 \times l}\quad \operatorname{vec}\left(\rho_B^i\right)^T\right)\right\|_{t r} \\\notag
			& =&\sqrt{l \beta^2+1}\operatorname{tr}\left(\left(\begin{array}{c}
				\alpha E_{l \times 1} \\
				\operatorname{vec}\left(\rho_A^i\right)
			\end{array}\right)\left(\alpha E_{1 \times l} \quad \operatorname{vec}\left(\rho_A^i\right)^T\right)\right)^{\frac{1}{2}} \\
			& =&\sqrt{l \beta^2+1} \sqrt{l \alpha^2+1}.
	\end{eqnarray}
Combining (\ref{eq:4}), (\ref{eq:6}) and (\ref{eq:27}) we get
	\begin{eqnarray}\notag
			&&\left\|\mathcal{M}_{\alpha, \beta}^{l}\left(\rho\right)\right\|_{t r}\\\notag& \leq& \sum_{i} p_{i}\left\|\mathcal{M}_{\alpha, \beta}^{l}\left(\rho_{A}^{i} \otimes \rho_{B}^{i}\right)\right\|_{t r}
			\\\notag&=&\sqrt{\left(l \alpha^{2}+1\right)\left(l \beta^{2}+1\right)},
	\end{eqnarray}
which completes the proof.
\end{proof}

We illustrate the Theorem \ref{th:1} by two examples.

\begin{example}\label{ex:1}
Consider the state $\rho_{x}=x\left|\xi \right\rangle\left\langle\xi\right|+ \left(1-x\right)\rho_{d}$, where $\rho_{d}$ is the $2 \times 4$ bound entangled state,
$$
	\rho_{d}=\frac{1}{1+7d}
	\begin{pmatrix}
		d & 0 & 0 & 0 & 0 & d & 0 & 0 \\
		0 & d & 0 & 0 & 0 & 0 & d & 0 \\
		0 & 0 & d & 0 & 0 & 0 & 0 & d \\
		0 & 0 & 0 & d & 0 & 0 & 0 & 0 \\
		0 & 0 & 0 & 0 & \frac{1+d}{2} & 0 & 0 & \frac{\sqrt{1-d^2}}{2} \\
		d & 0 & 0 & 0 & 0 & d & 0 & 0\\
		0 & d & 0 & 0 & 0 & 0 & d & 0\\
		0 & 0 & d & 0 & \frac{\sqrt{1-d^2}}{2} & 0 & 0 & \frac{1+d}{2} \\
	\end{pmatrix},
$$
where $0 \textless d\textless 1$, and $\left|\xi \right\rangle=\frac{1}{\sqrt{2}}(|00\rangle+|11\rangle)$.
\end{example}

Set $d=0.9$ and take $\alpha=11.66$, $\beta=11.75$ and $l=5$ in Theorem \ref{th:1}. By direct calculation we get from Theorem \ref{th:1} that $\rho_{x}$ is entangled for $0.233889 \leq x \leq 1$. While from the realignment criterion \cite{chen2002matrix}, the entanglement of $\rho_{x}$ is detected for $0.252758\leq x \leq 1$. The Theorem $1$ in \cite{shi2023family} detects the entanglement of $\rho_{x}$ for $0.233931 \leq x \leq 1$, and the Corollary 2.1 in \cite{shen2015separability} detects the entanglement of $\rho_{x}$ for $0.234778  \leq x \leq 1$. Obviously our Theorem \ref{th:1} detects better the entanglement of the state $\rho_{x}$.

\begin{example}\label{ex:2}
Consider the mixture of the bound entangled state proposed by Horodecki \cite{horodecki1997separability},
	$$
	\rho_{x}=\frac{1}{1+8x}
	\begin{pmatrix}
		x & 0 & 0 & 0 & x & 0 & 0 & 0 & x\\
		0 & x & 0 & 0 & 0 & 0 & 0 & 0 & 0\\
		0 & 0 & x & 0 & 0 & 0 & 0 & 0 & 0\\
		0 & 0 & 0 & x & 0 & 0 & 0 & 0 & 0 \\
		x & 0 & 0 & 0 & x & 0 & 0 & 0 &  x\\
		0 & 0 & 0 & 0 & 0 & x & 0 & 0 & 0\\
		0 & 0 & 0 & 0 & 0 & 0 & \frac{1+x}{2} & 0 & \frac{\sqrt{1-x^2}}{2}\\
		x & 0 & 0 & 0 & x & 0 & \frac{\sqrt{1-x^2}}{2} & 0 & \frac{1+x}{2}\\
	\end{pmatrix}
	$$
and the $9\times9$ identity matrix $I_9$,
$$
\rho_{p}=\frac{1-p}{9} I_{9}+p \rho_{x}.
$$
\end{example}

We take $\alpha=\beta=2$ and $l=10$. We compare among the results from our Theorem \ref{th:1}, realignment criterion from \cite{chen2002matrix} and Theorem $1$ in \cite{shi2023family} for different values of $x$, see Table \ref{tab:1}. Table \ref{tab:1} shows that our Theorem \ref{th:1} detects better the entanglement of the state $\rho_{p}$ than the criteria from \cite{chen2002matrix} and \cite{shi2023family}.
\setlength{\tabcolsep}{15pt}
\begin{center}
	\begin{table*}[ht]
\caption{Entanglement of the state $\rho_{p}$ in Example \ref{ex:2} for different values of $x$}
		\label{tab:1}
		\begin{tabular}{cccc}
			\hline\noalign{\smallskip}
		$x$ & realignment in \cite{chen2002matrix} & Theorem $1$ in \cite{shi2023family} & Our Theorem \ref{th:1}  \\
		\noalign{\smallskip}\hline\noalign{\smallskip}
		0.2 & $0.9955 \leq p \leq 1$ & $0.9943  \leq p \leq 1$ & $0.9942  \leq p \leq 1$ \\
		
		0.4 & $0.9960  \leq p \leq 1$ & $0.9948  \leq p \leq 1$ & $0.9947  \leq p \leq 1$ \\
		
		0.6 & $0.9972  \leq p \leq 1$ & $0.9964  \leq p \leq 1$ & $0.9963 \leq p \leq 1$ \\
		
		0.8 & $0.9986  \leq p \leq 1$ & $0.9982  \leq p \leq 1$ & $0.99815  \leq p \leq 1$ \\
		
		0.9 & $0.9993  \leq p \leq 1$ & $0.9991  \leq p \leq 1$ & $0.99908  \leq p \leq 1$ \\
		\noalign{\smallskip}\hline
		\end{tabular}
	\end{table*}
\end{center}

\subsection{Lower bounds of concurrence and CREN for bipartite states}
The concurrence of a pure state $|\varphi\rangle\in\mathbb{C}^{d_{A}} \otimes \mathbb{C}^{d_{B}}$ is defined by \cite{rungta2001universal}
\begin{eqnarray}\label{C:1}
	C\left(|\varphi\rangle\right)=\sqrt{2\left(1-\operatorname{tr} \left(\rho_{A}^{2}\right)\right)},
\end{eqnarray}
where $\rho_{A}=\operatorname{tr}_{B}\left(|\varphi\rangle\langle\varphi|\right)$. The concurrence of a mixed state $\rho$ is defined as
\begin{eqnarray}\label{C:2}
	C\left(\rho\right)=\min _{\left\{p_{i}, \left|\varphi_{i}\right\rangle\right\}} \sum_{i} p_{i} C\left(\left|\varphi_{i}\right\rangle\right),
\end{eqnarray}
where the minimum is taken over all possible ensemble decompositions of
$\rho=\sum\limits_{i}p_{i}\left|\varphi_{i}\right\rangle\left\langle\varphi_{i}\right|$, $p_{i} \geqslant 0$ with $\quad\sum\limits_{i}p_{i}=1$.

The CREN  of a pure state $|\varphi\rangle\in \mathbb{C}^{d_{A}} \otimes \mathbb{C}^{d_{B}}$ is defined by \cite{lee2003convex}
\begin{eqnarray}\notag
	 \mathcal{N}\left(|\varphi\rangle\right)
=\frac{\left\|\left(|\varphi\rangle\langle\varphi|\right)^{T_{B}}\right\|_{t r}-1}{d-1},
\end{eqnarray}
where $d=\min \left(d_A, d_B\right)$, $\left(|\varphi\rangle\langle\varphi|\right)^{T_{B}}$ denotes the partial transpose of $|\varphi\rangle\langle\varphi|$. For a mixed state $\rho$, its CREN is defined via convex roof extension,
\begin{eqnarray}\label{N:2}
	\mathcal{N}\left(\rho\right)=\min _{\left\{p_{i},\left|\varphi_{i}\right\rangle\right\}} \sum_{i} p_{i} \mathcal{N}\left(\left|\varphi_{i}\right\rangle\right),
\end{eqnarray}
where the minimum is taken over all possible pure state decompositions of $\rho$.

To derive the lower bounds of concurrence and CREN for arbitrary density matrices, we first present the following lemma, see proof in Appendix B.

\begin{lemma}\label{le:2} Let $|\varphi\rangle$ be a pure bipartite state in systems A and B, with Schmidt decomposition $|\varphi\rangle=\sum\limits_{i=0}^{d-1} \sqrt{\mu_{i}}\left|i_{A} i_{B}\right\rangle$, where $d=\min \left(d_{A}, d_{B}\right)$. Then

\noindent(1)
$\left\|\mathcal{M}_{\alpha, \beta}^{l}\left(|\varphi\rangle\langle\varphi|\right)\right\|_{t r}
\\\leq \sqrt{\left(l \alpha^{2}+1\right)\left(l \beta^{2}+1\right)} +2\sum\limits_{0 \leq i<j \leq d-1} \sqrt{\mu_{i} \mu_{j}}$,

\noindent(2)	
	$ \left\|\mathcal{M}_{\alpha, \beta}^{l}(|\varphi\rangle\langle\varphi|)\right\|_{t r}\leq \sqrt{\left(l \alpha^{2}+1\right)\left(l \beta^{2}+1\right)}+(d-1)$.
\end{lemma}

According to the above lemma, we have

\begin{theorem} \label{th:2}
For a bipartite state $\rho\in\mathbb{C}^{d_{A}} \otimes \mathbb{C}^{d_{B}}$, the concurrence satisfies that
\begin{eqnarray}\notag
C \!\left(\rho\right)\geq \frac{\sqrt{2}}{\sqrt{d(d-1)}}\!\!\left(\left\|\mathcal{M}_{\alpha, \beta}^{l}\left(\rho\right)\right\|_{t r}\!\!-\!\sqrt{\left(l \alpha^{2}+1\!\right)\left(l \beta^{2}+1\!\right)}\right),
\end{eqnarray}
where $d=\min \left(d_{A}, d_{B}\right)$.
\end{theorem}

\begin{proof}
For any $|\varphi\rangle\in\mathbb{C}^{d_{A}} \otimes \mathbb{C}^{d_{B}}$ with Schmidt form $|\varphi\rangle=\sum\limits_{i=0}^{d-1} \sqrt{\mu_{i}}\left|i_{A} i_{B}\right\rangle$, one has \cite{chen2005concurrence}
	\begin{eqnarray}\notag
		C^{2}(|\varphi\rangle) \geq \frac{8}{d(d-1)}\left(\sum_{0 \leq i<j \leq d-1} \sqrt{\mu_{i} \mu_{j}}\right)^{2}.
	\end{eqnarray}
	From (1) of Lemma \ref{le:2}, we get
	\begin{eqnarray}\notag
			&& \ \ \ \ C(|\varphi\rangle) \\\notag&&\geq \frac{2 \sqrt{2}}{\sqrt{d(d-1)}} \sum_{0 \leq i<j \leq d-1} \sqrt{\mu_{i} \mu_{j}} \\\notag
			&& \geq \!\!\frac{\sqrt{2}}{\sqrt{d(d-1)}}\!\left(\!\left\|\mathcal{M}_{\alpha, \beta}^{l}\!\left(|\varphi\rangle\langle\varphi|\right)\right\|_{t r}\!\!\!\!-\!\sqrt{\left(l \alpha^{2}\!+\!1\right)\!\left(l \beta^{2}\!+\!1\right)}\!\right).
	\end{eqnarray}

Let $\left\{p_{i},\left|\varphi_{i}\right\rangle\right\}$ be the optimal decomposition of $\rho$ such that $C\left(\rho\right)=\sum\limits_{i} p_{i} C\left(\left|\varphi_{i}\right\rangle\right)$. From Lemma 1 we have
	\begin{eqnarray}\notag
		\begin{aligned}
			C\left(\rho\right)&=\sum_{i} p_{i} C\left(\left|\varphi_{i}\right\rangle\right) \\
			&\geq \frac{\sqrt{2}}{\sqrt{d(d-1)}} \sum_{i} p_{i}\left(\left\|\mathcal{M}_{\alpha, \beta}^{l}\left(\left|\varphi_{i}\right\rangle\left\langle\varphi_{i}\right|\right)\right\|_{t r}\right.\\&\left.\quad-\sqrt{\left(l \alpha^{2}+1\right)\left(l \beta^{2}+1\right)}\right) \\
			& \geq \frac{\sqrt{2}}{\sqrt{d(d-1)}}\left(\left\|\mathcal{M}_{\alpha, \beta}^{l}\left(\rho\right)\right\|_{t r}\right.\\&\left.\quad-\sqrt{\left(l \alpha^{2}+1\right)\left(l \beta^{2}+1\right)}\right).
		\end{aligned}
	\end{eqnarray}
Therefore, Theorem \ref{th:2} holds.
\end{proof}

\begin{theorem}\label{th:3} For the CREN of any state $\rho\in\mathbb{C}^{d_{A}} \otimes \mathbb{C}^{d_{B}}$, we have	
\begin{eqnarray}\notag
		\mathcal{N}\left(\rho\right) \geq \frac{\left\|\mathcal{M}_{\alpha, \beta}^{l}\left(\rho\right)\right\|_{t r}-\sqrt{\left(l \alpha^{2}+1\right)\left(l \beta^{2}+1\right)}}{d-1},
	\end{eqnarray}
where $d=\min \left(d_{A}, d_{B}\right)$.
\end{theorem}

\begin{proof}
For any pure state $|\varphi\rangle\in\mathbb{C}^{d_{A}} \otimes \mathbb{C}^{d_{B}}$ with Schmidt form $|\varphi\rangle=\sum\limits_{i=0}^{d-1} \sqrt{\mu_{i}}\left|i_{A} i_{B}\right\rangle$, one has \cite{vidal2002computable}
	\begin{eqnarray}\notag
		\mathcal{N}\left(|\varphi\rangle\right)=\frac{2 \sum\limits_{0 \leq i<j \leq d-1} \sqrt{\mu_{i} \mu_{j}}}{d-1}.
	\end{eqnarray}
Using (1) of Lemma \ref{le:2}, we get
	\begin{eqnarray}\notag
		&& 2 \sum_{0 \leq i<j \leq d-1} \sqrt{\mu_{i} \mu_{j}} \\\notag&\geq&\left\|\mathcal{M}_{\alpha, \beta}^{l}\left(|\varphi\rangle\langle\varphi|\right)\right\|_{t r}-\sqrt{\left(l \alpha^{2}+1\right)\left(l \beta^{2}+1\right)}.
	\end{eqnarray}
Let $\left\{p_{i},\left|\varphi_{i}\right\rangle\right\}$ be the optimal pure state decomposition of $\rho$ such that $\mathcal{N}(\rho)=\sum\limits_{i} p_{i} \mathcal{N}\left(\left|\varphi_{i}\right\rangle\right)$. Then using Lemma 1 we obtain
	\begin{eqnarray}\notag
	\begin{aligned}
			\mathcal{N}\left(\rho\right)&=\sum_{i} p_{i} \mathcal{N}\left(\left|\varphi_{i}\right\rangle\right) \\
			&\geq \sum_{i} p_{i} \frac{\left\|\mathcal{M}_{\alpha, \beta}^{l}\left(\left|\varphi_{i}\right\rangle\left\langle\varphi_{i}\right|\right)\right\|_{t r}-\sqrt{\left(l \alpha^{2}+1\right)\!\left(l \beta^{2}+1\right)}}{d-1} \\
			&\geq \frac{\left\|\mathcal{M}_{\alpha, \beta}^{l}\left(\rho\right)\right\|_{t r}-\sqrt{\left(l \alpha^{2}+1\right)\left(l \beta^{2}+1\right)}}{d-1},
		\end{aligned}
	\end{eqnarray}
which completes the proof.
\end{proof}

\begin{example}
The following $3 \times 3$ $P P T$ entangled state was introduced in \cite{bennett1999unextendible},
$$
\rho=\frac{1}{4}\left(I_{9}-\sum_{i=0}^{4}
\left|\varphi_{i}\right\rangle\left\langle\varphi_{i}\right|\right),
$$
where
$$
\begin{aligned}
	& \left|\varphi_{0}\right\rangle=|0\rangle(|0\rangle-|1\rangle) / \sqrt{2} \\
	& \left|\varphi_{1}\right\rangle=(|0\rangle-|1\rangle)|2\rangle / \sqrt{2} \\
	& \left|\varphi_{2}\right\rangle=|2\rangle(|1\rangle-|2\rangle) / \sqrt{2} \\
	& \left|\varphi_{3}\right\rangle=(|1\rangle-|2\rangle)|0\rangle / \sqrt{2} \\
	& \left|\varphi_{4}\right\rangle=(|0\rangle+|1\rangle+|2\rangle)(|0\rangle+|1\rangle+|2\rangle) / 3.
\end{aligned}
$$
Let us consider the mixture of $\rho$ with white noise,
$$
\rho_{t}=\frac{1-t}{9} I_{9}+t \rho.
$$
\end{example}

In Fig.\ref{fig:1}, we take $\alpha=\beta=1$. The solid green line is the bound from Theorem \ref{th:2}, which shows that $\rho_{t}$ is entangled for $0.88248 \leq t \leq\ 1$. The dot dashed orange line is the bound from Theorem $6$ in \cite{shi2023family}, from which $\rho_{t}$ is entangled for $0.88438 \leq t \leq\ 1$. And the dashed blue line is the bound from Theorem in \cite{chen2005concurrence}, from which $\rho_{t}$ is entangled for $0.8897 \leq t \leq\ 1$. Clearly, the bound in Theorem \ref{th:2} is better than the bounds from \cite{shi2023family} and \cite{chen2005concurrence}.

In Fig.\ref{fig:2}, we take $l=2$ and consider different values of $\alpha$ and $\beta$. The solid orange line is the bound from Theorem \ref{th:2} with $\alpha=\beta=10$. The dashed blue line is the bound from Theorem \ref{th:2} with $\alpha=\beta=1$. It is seen that the lower bound of concurrence is better when $\alpha=\beta=10$.

In Fig.\ref{fig:3}, we consider the lower bound of CREN from Theorem \ref{th:3} for the example 3. We take $\alpha=\beta=1$. The solid orange line is the bound from Theorem \ref{th:3} with $l=10$. The dashed blue line is the bound from Theorem $7$ in \cite{shi2023family}. Clearly, the bound in Theorem \ref{th:3} is better than the bound in \cite{shi2023family}.

In Fig.\ref{fig:4}, we take $l=2$ and consider different values of $\alpha$ and $\beta$. The solid orange line is the bound from Theorem \ref{th:3} with $\alpha=\beta=10$. And the dashed blue line is the bound from Theorem \ref{th:3} with $\alpha=\beta=1$. It can be seen that the lower bound of CREN is better when $\alpha=\beta=10$.
\begin{figure}[h]
	\centering
	\includegraphics[scale=0.88]{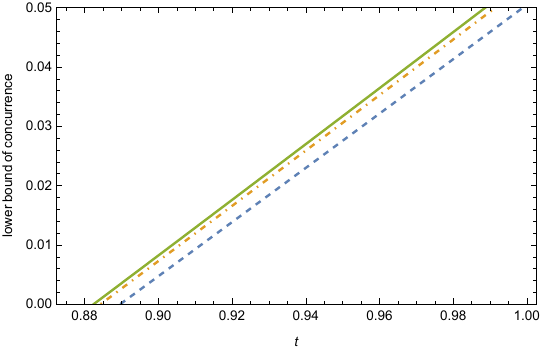}
	\caption{Lower bounds of concurrence for $\rho_{t}$ with $\alpha=\beta=1$, solid green line for the bound from Theorem \ref{th:2}, dot dashed orange line for the bound from Theorem $6$ in \cite{shi2023family}, dashed blue line for the bound from Theorem in \cite{chen2005concurrence}.}
	\label{fig:1}
\end{figure}
\begin{figure}[h]
	\centering
	\includegraphics[scale=0.88]{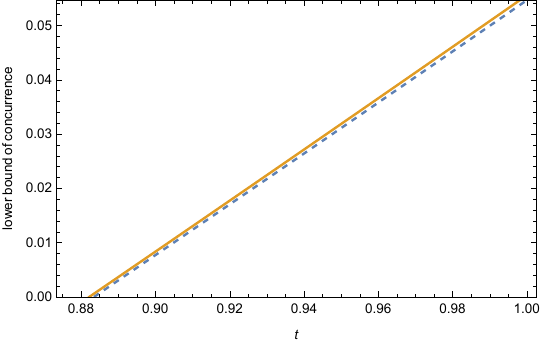}
	\caption{Lower bounds of concurrence for $\rho_{t}$ with $l=2$, solid orange (dashed blue) line for the bound from Theorem \ref{th:2} with $\alpha=\beta=10$ ($\alpha=\beta=1$).}
	\label{fig:2}
\end{figure}
\begin{figure}[h]
	\centering
	\includegraphics[scale=0.88]{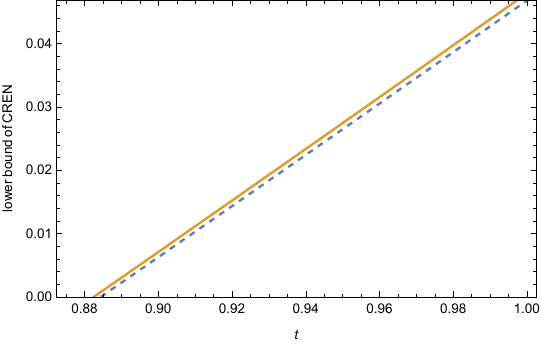}
	\caption{Lower bounds of CREN for $\rho_{t}$ with $\alpha=\beta=1$, solid orange line for the bound from Theorem \ref{th:3} with $l=10$, dashed blue line for the bound from Theorem $7$ in \cite{shi2023family}.}
	\label{fig:3}
\end{figure}
\begin{figure}[h]
	\centering
	\includegraphics[scale=0.88]{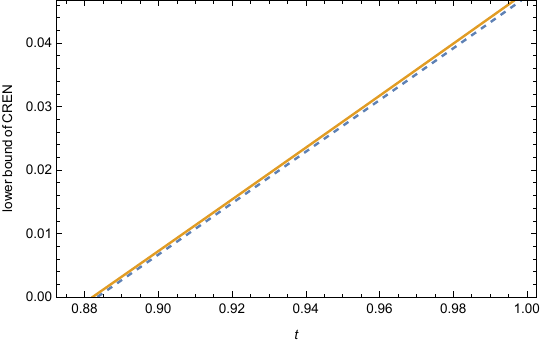}
	\caption{Lower bounds of CREN for $\rho_{t}$ with $l=2$, solid orange (dashed blue) line for the bound from Theorem \ref{th:3} with $\alpha=\beta=10$ ($\alpha=\beta=1$).}
	\label{fig:4}
\end{figure}

\section{Detection and measures of multipartite entanglement}\label{S:3}
\subsection{Separability criteria for multipartite states}

We first consider tripartite case. Denote the three bipartitions of a tripartite quantum state $\rho\in\mathbb{C}^{d_{1}} \otimes \mathbb{C}^{d_{2}} \otimes \mathbb{C}^{d_{3}}$ as $1|23$, $2|13$ and $3|12$. If a tripartite state $\rho$ is biseparable \cite{jing2022criteria}, then
\begin{eqnarray}\notag
	\begin{aligned}
	&\rho=\sum_i p_i|\varphi_i\rangle^{1|23}\langle\varphi_i|+\sum_j p_j|\varphi_j\rangle^{2|13}\langle\varphi_j|\\&\quad\quad+\sum_k p_k|\varphi_k\rangle^{3|12}\langle\varphi_k|,
	\end{aligned}
\end{eqnarray}
where $p_i, p_j, p_k \geq 0$, with $\sum\limits_i p_i+\sum\limits_j p_j+\sum\limits_k p_k=1$. Otherwise, $\rho$ is called genuinely tripartite entangled.
We define
\begin{eqnarray}\label{M:1} \mathscr{M}(\rho)\!\!=\!\frac{1}{3}\left(\left\|\mathscr{M}_{1|23}
\!\left(\rho\right)\!\right\|_{t r}\!\!+\!\left\|\mathscr{M}_{2|13}\!\left(\rho\right)\!\right\|_{t r}\!\!+\!\left\|\mathscr{M}_{3|12}\!\left(\rho\right)\!\right\|_{t r}\!\right),
\end{eqnarray}
where $\mathscr{M}_{i|jk}$ stands for the matrix (\ref{eq:m1}) under bipartition $i$ and $jk$, $\{i, j, k\}=\{1,2,3\}$.

\begin{theorem}\label{th:6}
If a tripartite state $\rho\in\mathbb{C}^{d_{1}} \otimes \mathbb{C}^{d_{2}} \otimes \mathbb{C}^{d_{3}}$ is biseparable, then
	\begin{eqnarray}\notag
		\mathscr{M}(\rho) \leq \sqrt{\left(l \alpha^{2}+1\right)\left(l \beta^{2}+1\right)}+\frac{2(d-1)}{3},
	\end{eqnarray}
where $d_i=d$ $\left(i=1,2,3\right)$.
\end{theorem}

\begin{proof}
Let $|\varphi\rangle\in\mathbb{C}^{d_{1}} \otimes \mathbb{C}^{d_{2}} \otimes \mathbb{C}^{d_{3}}$ be biseparable under bipartition $1|23$, i.e., $|\varphi\rangle=\left|\varphi_{1}\right\rangle \otimes\left|\varphi_{23}\right\rangle$. From Theorem \ref{th:1} and (2) of Lemma \ref{le:2} we have	
\begin{eqnarray}\label{M:2}
		\begin{aligned}
			& \quad\mathscr{M}(|\varphi\rangle\langle\varphi|)\\
			& =\frac{1}{3}\left(\left\|\mathscr{M}_{1|23}\left(|\varphi\rangle\langle\varphi|\right)\right\|_{t r}+\left\|\mathscr{M}_{2|13}\left(|\varphi\rangle\langle\varphi|\right)\right\|_{t r}\right.\\&\left.\quad+\left\|\mathscr{M}_{3|12}\left(|\varphi\rangle\langle\varphi|\right)\right\|_{t r}\right) \\
			& \leq \sqrt{\left(l \alpha^{2}+1\right)\left(l \beta^{2}+1\right)}+\frac{2(d-1)}{3},
		\end{aligned}
	\end{eqnarray}
which holds also for biseparable states under bipartitions $2|13$ and $3|12$.
	
For biseparable mixed state, $\rho=\sum_{i} p_{i}\left|\varphi_{i}\right\rangle\left\langle\varphi_{i}\right|$, $p_{i} \geq 0$, $\sum_{i} p_{i}=1$, we have
	\begin{eqnarray}\notag
		\mathscr{M}(\rho)\leq\sum_{i} p_{i}\mathscr{M}(|\varphi_{i}\rangle\langle\varphi_{i}|),
	\end{eqnarray}
which gives rise to that
	\begin{eqnarray}\notag
		\begin{aligned}
			\mathscr{M}(\rho)
			& \leq \sum_{i} p_{i}\left( \sqrt{\left(l \alpha^{2}+1\right)\left(l \beta^{2}+1\right)}+\frac{2(d-1)}{3}\right) \\
			& =\sqrt{\left(l \alpha^{2}+1\right)\left(l \beta^{2}+1\right)}+\frac{2(d-1)}{3}.
		\end{aligned}
	\end{eqnarray}
by using from (\ref{M:2}).
\end{proof}

We provide an example to illustrate the Theorem \ref{th:6}.
\begin{example}
Consider the state $\rho_{W}^{q}$ in $\mathbb{C}^3 \otimes \mathbb{C}^3 \otimes \mathbb{C}^3$,
	$$
	\rho_{W}^{q}=\frac{1-q}{27} I_{27}+q\left|\varphi_{W}\right\rangle\left\langle\varphi_{W}\right|,
	$$
	where $\left|\varphi_{W}\right\rangle=\frac{1}{\sqrt{6}}(|001\rangle+|010\rangle+|100\rangle+|112\rangle+|121\rangle+|211\rangle)$, $I_{27}$ is the $27 \times 27$ identity matrix.
\end{example}

In this case $d=3$. When $\alpha=\beta= 1$ and $l=2$, Theorem \ref{th:6} detects the genuine tripartite entanglement of $\rho_{W}^{q}$ for $0.805211 \leq q \leq 1$. When $\alpha=\beta= 1$ and $l=1$, Theorem \ref{th:6} reduces to the Theorem $2$ in \cite{qi2024detection}, which detects the genuine tripartite entanglement of $\rho_{W}^{q}$ for $0.805321 \leq q \leq 1$. The Theorem $1$ in \cite{de2011multipartite} detects the genuine tripartite entanglement of $\rho_{W}^{q}$ for $0.917663 \textless q \leq 1$. Obviously, our Theorem \ref{th:6} is more effective in detecting the genuine tripartite entanglement. $\\$

Next we consider the fully separability of general multipartite states. Any multipartite state $\rho\in\mathbb{C}^{d_{1}} \otimes \mathbb{C}^{d_{2}} \otimes \cdots \otimes \mathbb{C}^{d_{n}}$ can be written as
\begin{eqnarray}\notag
	\rho=\sum_{i} Y_{1}^{i} \otimes Y_{2}^{i} \otimes \cdots \otimes Y_{n}^{i},
\end{eqnarray}
where $Y_{j}^{i} \in \mathbb{C}^{d_{j} \times d_{j}}$, $j=1,2, \cdots, n$. We define
\begin{eqnarray}\label{eq:13}\notag
&&\mathcal{M} \mathcal{R}_{\alpha_{q}, \ldots, \alpha_{n}}^{l, q}(\rho)\\&=&\sum_{i} Y_{1}^{i} \otimes \!\cdots \!\otimes Y_{q-1}^{i} \otimes\! \mathcal{M}_{\alpha_{q}, \ldots, \alpha_{n}}^{l}\left[Y_{q}^{i}, \ldots, Y_{n}^{i}\right],
\end{eqnarray}
where $q=1,\ldots,n-1$, $\alpha_{q}, \ldots, \alpha_{n}$ are nonnegative real numbers, $l$ is a natural number,
\begin{eqnarray}\notag
	&&\mathcal{M}_{\alpha_{q}, \ldots, \alpha_{n}}^{l}\left[Y_{q}^{i}, \ldots, Y_{n}^{i}\right]\\\notag&=&\left(\begin{array}{c}
		\alpha_{q} E_{l \times 1} \\
		\operatorname{vec}\left(Y_{q}^{i}\right)
	\end{array}\right) \bigotimes_{j=0}^{n-q-1}\left(\alpha_{n-j} E_{1 \times l} \quad \operatorname{vec}\left(Y_{n-j}^{i}\right)^{T}\right),
\end{eqnarray}
with $E_{l \times 1} \in \mathbb{C}^{l}$ a vector with all elements being 1 and $E_{1 \times l}=E_{l \times 1}^{T}$.
We have the following separability criterion for multipartite states.

\begin{theorem}\label{th:4}
If a multipartite state  $\rho\in\mathbb{C}^{d_{1}} \otimes \mathbb{C}^{d_{2}} \otimes \cdots \otimes \mathbb{C}^{d_{n}}$ is fully separable, then for any $1\leq{q}\leq{n}-1$,
\begin{eqnarray}\notag
\left\|\mathcal{M} \mathcal{R}_{\alpha_{q}, \ldots, \alpha_{n}}^{l, q}(\rho)\right\|_{t r} \leq \prod_{k=q}^{n} \sqrt{l \alpha_{k}^{2}+1},
	\end{eqnarray}
where $\mathcal{M} \mathcal{R}_{\alpha_{q},  \ldots, \alpha_{n}}^{l, q}(\rho)$ is defined in (\ref{eq:13}).
\end{theorem}

\begin{proof}
Since any fully separable state $\rho$ can be written as
$\rho=\sum_{i} p_{i} \rho_{1}^{i} \otimes \rho_{2}^{i} \otimes \cdots \otimes \rho_{n}^{i}$,
where $p_{i} \in[0,1]$ with $\sum\limits_{i} p_{i}=1$, we have
\begin{eqnarray}\notag
\begin{aligned}
&\quad\left\|\mathcal{M} \mathcal{R}_{\alpha_{q}, \ldots, \alpha_{n}}^{l, q}(\rho)\right\|_{t r}\\&=\left\|\sum\limits_{i} p_{i} \rho_{1}^{i} \otimes \cdots \otimes \rho_{q-1}^{i} \otimes \mathcal{M}_{\alpha_{q}, \cdots, \alpha_{n}}^{l}\left[\rho_{q}^{i}, \ldots, \rho_{n}^{i}\right]\right\|_{t r}  \\
			&\leq \sum\limits_{i}p_{i}\left\|\rho_{1}^{i} \otimes \cdots \otimes \rho_{q-1}^{i} \otimes \mathcal{M}_{\alpha_{q}, \ldots, \alpha_{n}}^{l}\left[\rho_{q}^{i}, \ldots, \rho_{n}^{i}\right]\right\|_{t r}
			\\&=\sum\limits_{i} p_{i}\operatorname{tr}\left(\rho_{1}^{i}\right) \cdots\operatorname{tr}\left(\rho_{q-1}^{i}\right)\prod_{k=q}^{n}\left\|\left(\begin{array}{c}
				\alpha_{k} E_{l \times 1} \\
				\operatorname{vec}\left(\rho_{k}^{i}\right)
			\end{array}\right)\right\|_{t r}\\
			&=\sum\limits_{i} p_{i}\prod_{k=q}^{n} \sqrt{l \alpha_{k}^{2}+1}\\&=\prod_{k=q}^{n} \sqrt{l \alpha_{k}^{2}+1},
		\end{aligned}
	\end{eqnarray}
which proves Theorem \ref{th:4}.
\end{proof}

The following example illustrates the power of the Theorem \ref{th:4}.

\begin{example}\label{ex:5}
	Consider the quantum state $\rho_{G H Z}^{p}$  \cite{gittsovich2010multiparticle},
$$
\rho_{G H Z}^{p}=\frac{1-p}{8} I_{8}+p\left|\psi_{G H Z}\right\rangle\left\langle\psi_{G H Z}\right|,
$$
where $0 \leq p \leq 1$ and
$\left|\psi_{G H Z}\right\rangle=\frac{1}{x}(|000\rangle+\varepsilon|110\rangle+|111\rangle)$
with the real parameter $\varepsilon$ and the normalization factor $x$.
\end{example}	
	
Set $\alpha_{i}=10^{-1}$ $(i=1,2, 3)$, $l=2$ and $q=1$. We use Theorem \ref{th:4} to detect the entanglement of state $\rho_{G H Z}^{p}$ and compare the results with Theorem $2.1$ in \cite{zhang2017realignment} and Theorem $3.1$ in \cite{shen2022optimization}. From Table \ref{tab:4}, it is seen that our Theorem \ref{th:4} is more efficient than the Theorem $2.1$ in \cite{zhang2017realignment} and the Theorem $3.1$ in \cite{shen2022optimization} for different values of $\varepsilon$.
	\setlength{\tabcolsep}{15pt}
	\begin{center}
		\begin{table*}[ht]
\caption{Entanglement of state $\rho_{G H Z}^{p}$ in Example \ref{ex:5} with respect to different values of $\varepsilon$}
			\label{tab:4}
			\begin{tabular}{cccc}
				\hline\noalign{\smallskip}
			$\varepsilon$ &  Theorem $2.1$ in \cite{zhang2017realignment} & Theorem $3.1$ in \cite{shen2022optimization} & Theorem \ref{th:4} \\
			\noalign{\smallskip}\hline\noalign{\smallskip}
			$10^{-1}$ & $0.4118\leq p \leq 1$ & $0.4039 \leq p \leq 1$ & $0.4026 \leq p \leq 1$ \\
			
			1 & $0.4256 \leq p \leq 1$ & $0.4200 \leq p \leq 1$ & $0.4194 \leq p \leq 1$ \\
			
			10 & $0.7727 \leq p \leq 1$ & $0.7665 \leq p \leq 1$ & $0.7652 \leq p \leq 1$ \\
			\noalign{\smallskip}\hline
			\end{tabular}
		\end{table*}
	\end{center}

\subsection{Lower bounds of GME concurrence for tripartite systems}
The GME concurrence of a pure state $|\varphi\rangle\in\mathbb{C}^{d} \otimes \mathbb{C}^{d} \otimes \mathbb{C}^{d}$ is defined by \cite{ma2011measure},
\begin{eqnarray}\notag
	 C_{G M E}(|\varphi\rangle)=\sqrt{\!\min \left\{1-\operatorname{tr}\left(\rho_{1}^{2}\right), 1-\operatorname{tr}\left(\rho_{2}^{2}\right), 1-\operatorname{tr}\left(\rho_{3}^{2}\right)\right\}},
\end{eqnarray}
where $\rho_{i}$ is the reduced density matrix of subsystem $i$.
The GME concurrence of a mixed state $\rho$ is defined as
\begin{eqnarray}\label{GME:1}
	C_{G M E}(\rho)=\min _{\left\{p_{i},\left|\varphi_{i}\right\rangle\right\}} \sum_{i} p_{i} C_{G M E}\left(\left|\varphi_{i}\right\rangle\right),
\end{eqnarray}
where the minimum is taken over all possible ensemble decompositions of
$\rho=\sum\limits_{i}p_{i}\left|\varphi_{i}\right\rangle\left\langle\varphi_{i}\right|$, $p_{i} \geqslant 0$ with $\quad\sum\limits_{i}p_{i}=1$.
We have the following theorem about the lower bound of GME concurrence, see proof in Appendix C.

\begin{theorem}\label{th:7}
For any tripartite state $\rho\in\mathbb{C}^{d} \otimes \mathbb{C}^{d} \otimes \mathbb{C}^{d}$, we have
\begin{eqnarray}\notag
		\begin{aligned}
			C_{G M E}(\rho)&\geq \frac{1}{\sqrt{d(d-1)}}{\bigg(\mathscr{M}(\rho)\bigg.}\\&{\bigg.-\sqrt{\left(l \alpha^{2}+1\right)\left(l \beta^{2}+1\right)}-\frac{2(d-1)}{3}\bigg)},
		\end{aligned}
\end{eqnarray}
where $\mathscr{M}(\rho)$ is defined in (\ref{M:1}).
\end{theorem}

\begin{example}
Consider the following state $\rho_{x}\in\mathbb{C}^2 \otimes \mathbb{C}^2 \otimes \mathbb{C}^2$,
	$$
	\rho_{x}=\frac{x}{8} I_{8}+(1-x)\left|\phi_{G H Z}\right\rangle\left\langle\phi_{G H Z}\right|, 0 \leq x \leq 1,
	$$
where $\left|\phi_{G H Z}\right\rangle=\frac{1}{\sqrt{2}}(|000\rangle+|111\rangle)$, $I_{8}$ is the $8 \times 8$ identity matrix.
\end{example}

Set $\alpha=\beta= 1$ and $l=2$. From Theorem \ref{th:7} we get $C_{G M E}(\rho_{x}) \geq\frac{3\sqrt{2}\left(1-x\right)}{4}+\frac{\sqrt{5\left(x^2-2x+10\right)}}{4}
-\frac{11\sqrt{2}}{6}$. Fig.\ref{fig:5} shows that $\rho_{x}$ is genuine tripartite entangled for $0 \leq x \leq 0.192912$. When $\alpha=\beta= 1$ and $l=1$, Theorem \ref{th:7} reduces to the Theorem $3$ in \cite{qi2024detection}, which detects the genuine tripartite entanglement of $\rho_{x}$ for $0 \leq x \textless 0.192758$. The Corollary $1$ in \cite{zhao2023detecting} shows that $\rho_{x}$ is genuine tripartite entangled for $0 \leq x \textless 0.1919$. It can be seen that our Theorem \ref{th:7} detects better the genuine tripartite entanglement.
\begin{figure}[h]
	\centering
	\includegraphics[scale=0.88]{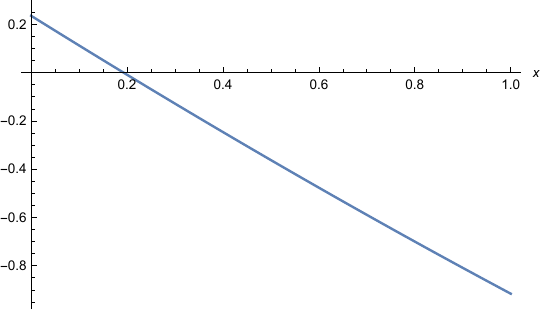}
	\caption{Lower bound of GME concurrence for $\rho_{x}$}
	\label{fig:5}
\end{figure}

\section{Conclusions and discussions}
The detection and estimation of quantum entanglement are of great significance in the quantum information processing. By constructing matrices based on the realignment of density matrices and the vectorization of the reduced density matrices, we have presented a family of separability criteria for both bipartite and multipartite systems. From these criteria we have derived new lower bounds of concurrence and convex-roof extended negativity. As for tripartite systems, we have also obtained the criteria to detect the genuine tripartite entanglement, and the lower bounds of the GME concurrence. We have shown by detailed examples that our separability criteria are more efficient than the known realignment criteria \cite{chen2002matrix} and the separability criterion given in \cite{shi2023family}. For multipartite cases, examples show that our criteria detect genuine tripartite entanglement and multipartite fully separable states better than the ones in \cite{qi2024detection,de2011multipartite,zhang2017realignment,shen2022optimization,zhao2023detecting}.
\bigskip
\section*{ACKNOWLEDGMENTS}
This work is supported by JCKYS2024604SSJS001, JCKYS2023604SSJS017, G2022180019L, the National Natural Science Foundation of China (NSFC) under Grants 12075159 and 12171044, and the specific research fund of the Innovation Platform for Academicians of Hainan Province under Grant No. YSPTZX202215.

\bigskip
\section*{APPENDIX}
\setcounter{equation}{0} \renewcommand%
\theequation{A\arabic{equation}}
\subsection{Proof of Lemma 1}
\begin{proof}
(1) The vectorization of matrices has the following properties,
\begin{eqnarray}\label{eq:v2}
\operatorname{vec}\left(k_{1} A+k_{2} B\right)=k_{1} \operatorname{vec}(A)+k_{2}\operatorname{vec}(B)
\end{eqnarray}
for any $A, B \in \mathbb{C}^{m \times n}$, $k_{i} \in \mathbb{R}$ $(i=1,2)$ and
\begin{eqnarray}\label{eq:v3}
\operatorname{vec}(A B C)=\left(C^{T} \otimes A\right) \operatorname{vec}(B)
\end{eqnarray}
for any $A \in \mathbb{C}^{m \times n}$, $B \in \mathbb{C}^{n \times p}$ and $C \in \mathbb{C}^{p \times q}$.

From (\ref{eq:v2}), it yields that for any $A, B \in \mathbb{C}^{mn \times mn}$, $k_{i} \in \mathbb{R}$ $(i=1,2)$,
\begin{eqnarray}\label{eq:r2}
\mathcal{R}\left(k_{1} A+k_{2} B\right)=k_{1} \mathcal{R}(A)+k_{2} \mathcal{R}(B).
\end{eqnarray}

Clearly, from (\ref{eq:v2}) for any $A, B \in \mathbb{C}^{m \times n}$, $k_{i} \in \mathbb{R}$ $(i=1,2)$, we have
\begin{eqnarray}\label{eq:w2}
	\omega_{l}\left(k_{1} A+k_{2} B\right)=k_{1} \omega_{l}(A)+k_{2} \omega_{l}(B).
\end{eqnarray}

From the definition of $\mathcal{M}_{\alpha,\beta}^{l}\left(\rho\right)$ in (\ref{eq:m1}), (\ref{eq:r2}) and (\ref{eq:w2}), we get
	\begin{eqnarray}\notag
			&& \mathcal{M}_{\alpha, \beta}^{l}\left(\sum_{i=1}^{n} k_{i} \rho_{i}\right)\\ &=&\left(\begin{array}{cc}
				\alpha \beta E_{l \times l} & \alpha \omega_{l}\left(\operatorname{tr}_{A}\left(\sum\limits_{i=1}^{n} k_{i} \rho_{i}\right)\right)^{T} \\
				\beta \omega_{l}\left(\operatorname{tr}_{B}\left(\sum\limits_{i=1}^{n} k_{i} \rho_{i}\right)\right) & \mathcal{R}\left(\sum\limits_{i=1}^{n} k_{i} \rho_{i}\right)
			\end{array}\right)\notag \\
			& =&\left(\begin{array}{cc}
				\sum\limits_{i=1}^{n} k_{i} \alpha \beta E_{l \times l} & \alpha \omega_{l}\left(\sum\limits_{i=1}^{n} k_{i}\operatorname{tr}_{A}\left(\rho_{i}\right)\right)^{T} \\
				\beta \omega_{l}\left(\sum\limits_{i=1}^{n} k_{i} \operatorname{tr}_{B}\left(\rho_{i}\right)\right) & \sum\limits_{i=1}^{n} k_{i} \mathcal{R}\left(\rho_{i}\right)
			\end{array}\right)\notag \\
			& =&\left(\begin{array}{cc}
				\sum\limits_{i=1}^{n} k_{i} \alpha \beta E_{l \times l} & \alpha \sum\limits_{i=1}^{n} k_{i} \omega_{l}\left(\operatorname{tr}_{A}\left(\rho_{i}\right)\right)^{T} \\
				\beta \sum\limits_{i=1}^{n} k_{i} \omega_{l}\left(\operatorname{tr}_{B}\left(\rho_{i}\right)\right) & \sum\limits_{i=1}^{n} k_{i} \mathcal{R}\left(\rho_{i}\right)
			\end{array}\right)\notag \\
			& =&\sum\limits_{i=1}^{n} k_{i}\left(\begin{array}{cc}
				\alpha \beta E_{l \times l} & \alpha \omega_{l}\left(\operatorname{tr}_{A}\left(\rho_{i}\right)\right)^{T} \\
				\beta \omega_{l}\left(\operatorname{tr}_{B}\left(\rho_{i}\right)\right) & \mathcal{R}\left(\rho_{i}\right)
			\end{array}\right)\notag \\
			& =&\sum\limits_{i=1}^{n} k_{i} \mathcal{M}_{\alpha, \beta}^{l}\left(\rho_{i}\right)\notag,
	\end{eqnarray}
which proves (1) of Lemma \ref{P:1}.
	
(2) A general state $\rho$ can be written as \cite{rudolph2005further} $\rho=\sum\limits_{i} \xi_{i} \otimes \eta_{i}$, where $\xi_{i} \in \mathbb{C}^{d_{A} \times d_{A}}$ and $\eta_{i} \in \mathbb{C}^{d_{B} \times d_{B}}$. Denote $\sigma=(U \otimes V) \rho\left(U \otimes V\right)^{\dagger}$. We have
	\begin{eqnarray}\label{eq:m5}
		\mathcal{M}_{\alpha, \beta}^{l}(\sigma)=\left(\begin{array}{cc}
			\alpha \beta E_{l \times l} & \alpha \omega_{l}\left(\operatorname{tr}_{A}(\sigma)\right)^{T} \\
			\beta \omega_{l}\left(\operatorname{tr}_{B}(\sigma)\right) & \mathcal{R}(\sigma)
		\end{array}\right),
	\end{eqnarray}
	where
\begin{eqnarray}\notag
		&& \alpha \omega_{l}\left(\operatorname{tr}_{A}(\sigma)\right)^{T}
		\\ \notag& =&\alpha \omega_{l}\left(\operatorname{tr}_{A}\left((U \otimes V) \sum_{i}\left(\xi_{i} \otimes \eta_{i}\right)\left(U \otimes V\right)^{\dagger}\right)\right)^{T} \\\notag
		& =&\alpha \omega_{l}\left(\sum_{i} \operatorname{tr}_{A}\left(U \xi_{i} U^{\dagger} \otimes V \eta_{i} V^{\dagger}\right)\right)^{T} \\\notag
		& =&\alpha \omega_{l}\left(\sum_{i} \operatorname{tr}\left(\xi_{i}\right) V \eta_{i} V^{\dagger}\right)^{T}.
\end{eqnarray}
From (\ref{eq:v3}) and (\ref{eq:w2}), it yields that
	\begin{eqnarray}\label{eq:m7}\notag
			&& \alpha \omega_{l}\left(\operatorname{tr}_{A}(\sigma)\right)^{T}\\\notag&=&\alpha \sum_{i} \operatorname{tr}\left(\xi_{i}\right) \omega_{l}\left(V \eta_{i} V^{\dagger}\right)^{T} \\\notag
			& =&\alpha \sum_{i} \operatorname{tr}\left(\xi_{i}\right)\left(\operatorname{vec}\left(V \eta_{i} V^{\dagger}\right), \ldots, \operatorname{vec}\left(V \eta_{i} V^{\dagger}\right)\right)^{T} \\\notag
			& =&\alpha \sum_{i} \operatorname{tr}\left(\xi_{i}\right)\left(\left(\left(V^{\dagger}\right)^{T} \otimes V\right) \omega_{l}\left(\eta_{i}\right)\right)^{T} \\\notag
			& =&\alpha \sum_{i} \operatorname{tr}\left(\xi_{i}\right)\left((\overline{V} \otimes V) \omega_{l}\left(\eta_{i}\right)\right)^{T} \\\notag
			& =&\alpha \sum_{i} \operatorname{tr}\left(\xi_{i}\right) \omega_{l}\left(\eta_{i}\right)^{T}(\overline{V} \otimes V)^{T} \\\notag
			& =&\alpha \omega_{l}\left(\sum_{i} \operatorname{tr}_{A}\left(\xi_{i} \otimes \eta_{i}\right)\right)^{T}(\overline{V} \otimes V)^{T} \\
			& =&\alpha \omega_{l}\left(\operatorname{tr}_{A}(\rho)\right)^{T}(\overline{V} \otimes V)^{T}.
	\end{eqnarray}
Similarly, we have
	\begin{eqnarray}\label{eq:m8}
		\beta \omega_{l}\left(\operatorname{tr}_{B}(\sigma)\right)=\beta(\overline{U} \otimes U) \omega_{l}\left(\operatorname{tr}_{B}(\rho)\right).
	\end{eqnarray}

Using (\ref{eq:v3}), (\ref{eq:r2}) and (\ref{eq:r4}) we obtain
	\begin{eqnarray}\label{eq:m9}\notag
			&&\mathcal{R}(\sigma)\\\notag
			& =&\sum_{i} \mathcal{R}\left((U \otimes V)\left(\xi_{i} \otimes \eta_{i}\right)\left(U \otimes V\right)^{\dagger}\right) \\\notag
			& =&\sum_{i} \mathcal{R}\left(U \xi_{i} U^{\dagger} \otimes V \eta_{i} V^{\dagger}\right) \\\notag
			& =&\sum_{i} \operatorname{vec}\left(U \xi_{i} U^{\dagger}\right) \operatorname{vec}\left(V \eta_{i} V^{\dagger}\right)^{T} \\\notag
			& =&\sum_{i}\left(\left(U^{\dagger}\right)^{T} \otimes U\right) \operatorname{vec}\!\left(\xi_{i}\right)\!\left(\!\left(\left(V^{\dagger}\right)^{T} \otimes V\right) \operatorname{vec}\!\left(\!\eta_{i}\!\right)\!\right)^{T} \\\notag
			& =&\sum_{i}(\overline{U} \otimes U) \operatorname{vec}\left(\xi_{i}\right) \operatorname{vec}\left(\eta_{i}\right)^{T}(\overline{V} \otimes V)^{T} \\\notag
			& =&\sum_{i}(\overline{U} \otimes U) \mathcal{R}\left(\xi_{i} \otimes \eta_{i}\right)(\overline{V} \otimes V)^{T} \\
			& =&(\overline{U} \otimes U) \mathcal{R}(\rho)(\overline{V} \otimes V)^{T}.
	\end{eqnarray}
Therefore, combing (\ref{eq:m5}), (\ref{eq:m7}), (\ref{eq:m8}) and (\ref{eq:m9}) we have
	\begin{eqnarray}\notag
			&&\mathcal{M}_{\alpha, \beta}^{l}(\sigma)\\&=&\left(\!\!\begin{array}{cc}
				\alpha \beta E_{l \times l} & \alpha \omega_{l}\left(\operatorname{tr}_{A}(\rho)\right)^{T}(\overline{V} \otimes V)^{T} \\
				\beta(\overline{U} \otimes U) \omega_{l}\left(\operatorname{tr}_{B}(\rho)\right) & (\overline{U} \otimes U) \mathcal{R}(\rho)(\overline{V} \otimes V)^{T}
			\end{array}\!\!\right)\notag\\\notag
			& =&\widetilde{U}\left(\begin{array}{cc}
				\alpha \beta E_{l \times l} & \alpha \omega_{l}\left(\operatorname{tr}_{A}(\rho)\right)^{T} \\
				\beta \omega_{l}\left(\operatorname{tr}_{B}(\rho)\right) & \mathcal{R}(\rho)
			\end{array}\right)\widetilde{V}\\\notag
			&=&\widetilde{U} \mathcal{M}_{\alpha, \beta}^{l}(\rho) \widetilde{V},
	\end{eqnarray}
where $\overline{U}=\left(U^{\dagger}\right)^{T}$,
$$
\widetilde{U}=\left(\begin{array}{cc}I_{l \times l} & 0 \\ 0 & \overline{U} \otimes U\end{array}\right),~~ \widetilde{V}=\left(\begin{array}{cc}I_{l \times l} & 0 \\ 0 & V \otimes \overline{V}\end{array}\right)^{\dagger}.
$$
Hence,
	\begin{eqnarray}\notag
			&&\left\|\mathcal{M}_{\alpha, \beta}^{l}\left((U \otimes V) \rho\left(U \otimes V\right)^{\dagger}\right)\right\|_{t r}\\\notag&=&\left\|\widetilde{U} \mathcal{M}_{\alpha, \beta}^{l}(\rho) \widetilde{V}\right\|_{t r}\\\notag&=&\left\|\mathcal{M}_{\alpha, \beta}^{l}(\rho)\right\|_{t r},
	\end{eqnarray}
which proves (2) of Lemma \ref{P:1}.
\end{proof}

\subsection{Proof of Lemma 2}
\begin{proof}
	$\left(1\right)$ Since
$|\varphi\rangle=\sum\limits_{i=0}^{d-1} \sqrt{\mu_{i}}\left|i_{A} i_{B}\right\rangle$, one has
$$|\varphi\rangle\langle\varphi|=\sum\limits_{i=0}^{d-1} \sum\limits_{j=0}^{d-1} \sqrt{\mu_{i} \mu_{j}}\left|i_{A} i_{B}\right\rangle\left\langle j_{A} j_{B}\right|.$$
Hence,
	\begin{eqnarray}\notag
		\mathcal{M}_{\!\alpha, \beta}^{l}(|\varphi\rangle\langle\varphi|)\!=\left(\begin{array}{cc}
			\!\!\alpha \beta E_{l \times l} & \alpha \omega_{l}\left(\operatorname{tr}_{A}(|\varphi\rangle\langle\varphi|)\right)^{T} \\
			\!\!\beta \omega_{l}\left(\operatorname{tr}_{B}(|\varphi\rangle\langle\varphi|)\right) & \!\!\mathcal{R}(|\varphi\rangle\langle\varphi|)
		\end{array}\right).
	\end{eqnarray}

Let
	\begin{eqnarray}\notag
		\mathrm{M}_{1}=\left(\begin{array}{cc}
			\alpha \beta E_{l \times l} & \alpha A \\
			\beta A^{T} & \Delta_{1}
		\end{array}\right),~~ \mathrm{M}_{2}=\left(\begin{array}{cc}
			0 & 0 \\
			0 & \Delta_{2}
		\end{array}\right),
	\end{eqnarray}
	where
	\begin{widetext}
		\begin{eqnarray}\notag
			A=\left(\begin{array}{c}
				\mu_{0}, \underbrace{0, \ldots, 0}_{d-1}, 0, \mu_{1}, \underbrace{0, \ldots, 0}_{d-2}, 0,0, \mu_{2}, \underbrace{0, \ldots, 0}_{d-3}, \ldots, 0, \ldots, 0, \mu_{d-1} \\
				\vdots \\
				\mu_{0}, \underbrace{0, \ldots, 0}_{d-1}, 0, \mu_{1}, \underbrace{0, \ldots, 0}_{d-2}, 0,0, \mu_{2}, \underbrace{0, \ldots, 0}_{d-3}, \ldots, 0, \ldots, 0, \mu_{d-1}
			\end{array}\right)_{l \times d^{2}},
		\end{eqnarray}
		\begin{eqnarray}\notag
			\Delta_{1}=\operatorname{diag}(\mu_{0}, \underbrace{0, \ldots, 0}_{d-1}, 0, \mu_{1}, \underbrace{0, \ldots, 0}_{d-2}, 0,0, \mu_{2}, \underbrace{0, \ldots, 0}_{d-3}, \ldots, 0, \ldots, 0, \mu_{d-1})\in \mathbb{C}^{d^{2} \times d^{2}},
		\end{eqnarray}
		\begin{eqnarray}\notag
			\begin{aligned}
				& \Delta_{2}=\operatorname{diag}\left(0, \sqrt{\mu_{0} \mu_{1}}, \ldots, \sqrt{\mu_{0} \mu_{d-1}}, \sqrt{\mu_{1} \mu_{0}}, 0, \sqrt{\mu_{1} \mu_{2}}, \ldots, \sqrt{\mu_{1} \mu_{d-1}}\right. \\
				& \left.\sqrt{\mu_{2} \mu_{0}}, \sqrt{\mu_{2} \mu_{1}}, 0, \sqrt{\mu_{2} \mu_{3}}, \ldots, \sqrt{\mu_{2} \mu_{d-1}}, \ldots, \sqrt{\mu_{d-1} \mu_{0}}, \ldots, \sqrt{\mu_{d-1} \mu_{d-2}}, 0\right)\in \mathbb{C}^{d^{2} \times d^{2}}.
			\end{aligned}
		\end{eqnarray}
	\end{widetext}
Using the definition of trace norm we get
	\begin{eqnarray}\notag
		\left\|\mathcal{M}_{\alpha, \beta}^{l}(|\varphi\rangle\langle\varphi|)\right\|_{t r}=\left\|\mathrm{M}_{1}\right\|_{t r}+\left\|\mathrm{M}_{2}\right\|_{t r},
	\end{eqnarray}
	where $\left\|\mathrm{M}_{2}\right\|_{t r}=2 \sum\limits_{0 \leq i<j \leq d-1} \sqrt{\mu_{i} \mu_{j}}$. Since $\Delta_{1}=\mathcal{R}\left(\Delta_{1}\right)$ and $\Delta_{1}$ is a separable state, according to Theorem \ref{th:1} we get
	\begin{eqnarray}\notag
		\left\|\mathrm{M}_{1}\right\|_{t r}=\left\|\mathcal{M}_{\alpha, \beta}^{l}\left(\Delta_{1}\right)\right\|_{t r} \leq \sqrt{\left(l \alpha^{2}+1\right)\left(l \beta^{2}+1\right)}.
	\end{eqnarray}
	Therefore,
	\begin{eqnarray}\notag
			&& \left\|\mathcal{M}_{\alpha, \beta}^{l}(|\varphi\rangle\langle\varphi|)\right\|_{t r} \\ \notag&\leq& \sqrt{\left(l \alpha^{2}+1\right)\left(l \beta^{2}+1\right)}+2 \sum_{0 \leq i<j \leq d-1} \sqrt{\mu_{i} \mu_{j}}.
	\end{eqnarray}

$\left(2\right)$
It has been proved in \cite{qi2024detection} that for any $\mu_{i}>0$, $i=0,1, \ldots, d-1$, such that $\mu_{0}+\mu_{1}+\cdots+\mu_{d-1}=1$, one has
\begin{eqnarray}\notag
		2 \sum\limits_{0 \leq i<j \leq d-1} \sqrt{\mu_{i} \mu_{j}} \leq d-1.
\end{eqnarray}
Based on above relation and (1) of Lemma \ref{le:2}, we get
	\begin{eqnarray}\notag
			&&\left\|\mathcal{M}_{\alpha, \beta}^{l}(|\varphi\rangle\langle\varphi|)\right\|_{t r} \\\notag&
			\leq &\sqrt{\left(l \alpha^{2}+1\right)\left(l \beta^{2}+1\right)}+2 \sum_{0 \leq i<j \leq d-1} \sqrt{\mu_{i} \mu_{j}} \\\notag
			& \leq& \sqrt{\left(l \alpha^{2}+1\right)\left(l \beta^{2}+1\right)}+(d-1).
	\end{eqnarray}
Therefore, we complete the proof of (2) in Lemma \ref{le:2}.
\end{proof}

\subsection{Proof of Theorem 6}
\begin{proof}
	For a pure state $|\varphi\rangle$ in $\mathbb{C}^{d_{1}} \otimes \mathbb{C}^{d_{2}} \otimes \mathbb{C}^{d_{3}}$, using Theorem \ref{th:2} and (\ref{C:1}) we get
\begin{eqnarray}\label{B:1}
\begin{aligned}
			&\sqrt{1-\operatorname{tr}\left(\rho_{1}^{2}\right)} \geq \frac{1}{\sqrt{d(d-1)}}{\Big(\left\|\mathscr{M}_{1|23}\left(|\varphi\rangle\langle\varphi|\right)\right\|\bigg.}\\&{\Big.\qquad\quad\qquad\qquad-\sqrt{\left(l \alpha^{2}+1\right)\left(l \beta^{2}+1\right)}\Big)},
		\end{aligned}
\end{eqnarray}
\begin{eqnarray}\label{B:2}
		\begin{aligned}
			&\sqrt{1-\operatorname{tr}\left(\rho_{2}^{2}\right)} \geq \frac{1}{\sqrt{d(d-1)}}{\Big(\left\|\mathscr{M}_{2|13}\left(|\varphi\rangle\langle\varphi|\right)\right\|\Big.}\\&{\Big.\qquad\quad\qquad\qquad-\sqrt{\left(l \alpha^{2}+1\right)\left(l \beta^{2}+1\right)}\Big)}
		\end{aligned}
	\end{eqnarray}
and	
\begin{eqnarray}\label{B:3}
		\begin{aligned}
			&\sqrt{1-\operatorname{tr}\left(\rho_{3}^{2}\right)} \geq \frac{1}{\sqrt{d(d-1)}}{\Big(\left\|\mathscr{M}_{3|12}\left(|\varphi\rangle\langle\varphi|\right)\right\|\Big.}\\&{\Big.\qquad\quad\qquad\qquad-\sqrt{\left(l \alpha^{2}+1\right)\left(l \beta^{2}+1\right)}\Big)}.
		\end{aligned}
	\end{eqnarray}
Combining (\ref{B:1}), (\ref{B:2}) and (\ref{B:3}) we get
	\begin{eqnarray}\notag
		\begin{aligned}
			& \quad3 \sqrt{d(d-1)} \sqrt{1-\operatorname{tr}\left(\rho_{1}^{2}\right)}-3\mathscr{M}(|\varphi\rangle\langle\varphi|) \\&\quad+\left(3 \sqrt{\left(l \alpha^{2}+1\right)\left(l \beta^{2}+1\right)}+2(d-1)\right) \\
			& =3 \sqrt{d(d-1)} \sqrt{1-\operatorname{tr}\left(\rho_{1}^{2}\right)}-\left(\left\|\mathscr{M}_{1|23}\left(|\varphi\rangle\langle\varphi|\right)\right\|_{t r}\right.\\&\left.\quad+\left\|\mathscr{M}_{2|13}\left(|\varphi\rangle\langle\varphi|\right)\right\|_{t r}+\left\|\mathscr{M}_{3|12}\left(|\varphi\rangle\langle\varphi|\right)\right\|_{t r}\right) \\
			& \quad+\left(3 \sqrt{\left(l \alpha^{2}+1\right)\left(l \beta^{2}+1\right)}+2(d-1)\right) \\
			& \geq 3 \sqrt{d(d-1)} \sqrt{1-\operatorname{tr}\left(\rho_{1}^{2}\right)}\!\!-\!\sqrt{d(d-1)}\left(\sqrt{1-\operatorname{tr}\left(\rho_{1}^{2}\right)}\right.\\&\left.\quad+\sqrt{1-\operatorname{tr}\left(\rho_{2}^{2}\right)}+\sqrt{1-\operatorname{tr}\left(\rho_{3}^{2}\right)}\right)+2(d-1) \\
			& =2 \sqrt{d(d-1)} \sqrt{1-\operatorname{tr}\left(\rho_{1}^{2}\right)}-\sqrt{d(d-1)}
\left(\sqrt{1-\operatorname{tr}\left(\rho_{2}^{2}\right)}
\right.\\&\left.\quad+\sqrt{1-\operatorname{tr}\left(\rho_{3}^{2}\right)}\right)+2(d-1).
		\end{aligned}
	\end{eqnarray}

Since $\sqrt{1-\operatorname{tr}\left(\rho_{1}^{2}\right)} \geq 0$ and $\sqrt{1-\operatorname{tr}\left(\rho_{i}^{2}\right)} \leq \sqrt{1-\frac{1}{d}}(i=2,3)$, we obtain
\begin{eqnarray}\notag
		\begin{aligned}
			& \quad 3 \sqrt{d(d-1)} \sqrt{1-\operatorname{tr}\left(\rho_{1}^{2}\right)}-3\mathscr{M}(|\varphi\rangle\langle\varphi|)\\&\quad+\left(3 \sqrt{\left(l \alpha^{2}+1\right)\left(l \beta^{2}+1\right)}+2(d-1)\right) \\
			& \geq2 \sqrt{d(d-1)} \sqrt{1-\operatorname{tr}\left(\rho_{1}^{2}\right)}\!-\!\!\sqrt{d(d-1)}\!\left(\sqrt{1-\operatorname{tr}\left(\rho_{2}^{2}\right)}\right.\\&\left.\quad+\sqrt{1-\operatorname{tr}\left(\rho_{3}^{2}\right)}\right)+2(d-1) \\
			& \geq 2(d-1)-2 \sqrt{d(d-1)} \sqrt{1-\frac{1}{d}}
\\
            & =0.
		\end{aligned}
	\end{eqnarray}
Therefore
	\begin{eqnarray}\notag
		\begin{aligned}
			&\sqrt{1-\operatorname{tr}\left(\rho_{1}^{2}\right)} \geq \frac{1}{\sqrt{d(d-1)}}{\bigg(\mathscr{M}(|\varphi\rangle\langle\varphi|)\bigg.}\\&{\bigg.\qquad\quad\qquad\qquad-\sqrt{\left(l \alpha^{2}+1\right)\left(l \beta^{2}+1\right)}-\frac{2(d-1)}{3}\bigg)}.
		\end{aligned}
	\end{eqnarray}
	Similarly, for $i=2,3$, we have
	\begin{eqnarray}\notag
		\begin{aligned}
			&\sqrt{1-\operatorname{tr}\left(\rho_{i}^{2}\right)} \geq \frac{1}{\sqrt{d(d-1)}}{\bigg(\mathscr{M}(|\varphi\rangle\langle\varphi|)\bigg.}\\&{\bigg.\qquad\quad\qquad\qquad-\sqrt{\left(l \alpha^{2}+1\right)\left(l \beta^{2}+1\right)}-\frac{2(d-1)}{3}\bigg)}.
		\end{aligned}
	\end{eqnarray}
	
Now let $\left\{p_{i},\left|\varphi_{i}\right\rangle\right\}$ be the optimal decomposition for $\rho$ such that $C_{G M E}(\rho)=\sum\limits_{i} p_{i} C_{G M E}\left(\left|\varphi_{i}\right\rangle\right)$. Then
	\begin{eqnarray}\notag
		\begin{aligned}
			C_{G M E}(\rho)&=\sum\limits_{i} p_{i} C_{G M E}\left(\left|\varphi_{i}\right\rangle\right) \\
			& \geq \frac{1}{\sqrt{d(d-1)}} \sum\limits_{i} p_{i}{\bigg(\mathscr{M}(|\varphi_{i}\rangle\langle\varphi_{i}|)\bigg.}\\&{\bigg.\quad-\sqrt{\left(l \alpha^{2}+1\right)\left(l \beta^{2}+1\right)}-\frac{2(d-1)}{3}\bigg)} \\
			& \geq \frac{1}{\sqrt{d(d-1)}}{\bigg(\mathscr{M}(\rho)\bigg.}\\&{\bigg.\quad-\sqrt{\left(l \alpha^{2}+1\right)\left(l \beta^{2}+1\right)}-\frac{2(d-1)}{3}\bigg)},
		\end{aligned}
	\end{eqnarray}
which completes proof of Theorem \ref{th:7}.
\end{proof}

\end{document}